%% file: main.tex
\documentclass{IOS-Book-Article}

\usepackage{enumerate}
 \input{preamble}
\usepackage{todonotes}

\newcommand{\tc}{\textit{tc}}
\newcommand{\Ns}{\textit{Ns}}
\newcommand{\activekw}{\textit{active}}

\newcommand{\timedtracetype}{\textit{TimedTrace}}

\newcommand{\pers}{\textit{pers}}
\newcommand{\eph}{\textit{eph}}

\newcommand{\reachable}{\textit{reachable}}

\begin{document}

\begin{frontmatter}              %
\title{Sound Conflict Analysis for Timed Contract Automata}
\author[A]{\fnms{Shaun} \snm{Azzopardi}} %
and
\author[B]{\fnms{Gordon J.} \snm{Pace}}

\runningauthor{S. Azzopardi et al.}
\address[A]{Independent Researcher}
\address[B]{University of Malta, Malta}

\begin{abstract}
One can find various temporal deontic logics in literature, most focusing on discrete time. The literature on real-time constraints and deontic norms is much sparser. Thus, many analysis techniques which have been developed for deontic logics have not been considered for continuous time. In this paper we focus on the notion of conflict analysis which has been extensively studied for discrete time deontic logics. We present a sound, but not complete algorithm for detecting conflicts in timed contract automata and prove the correctness of the algorithm, illustrating the analysis on a case study.
\end{abstract}

\end{frontmatter}

\section{Introduction}\label{s:intro}
In the literature, one can find various variations of temporal deontic logics~\cite{governatori2007characterising}, however, much of it is limited to discrete model of time. Although one may argue that, with sufficiently fine temporal granularity, discrete time logics suffice for most intents and purposes, specifications and legal clauses may include real-time notions, which may not be easily mappable to a discrete setting.

We find some deontic logics able to deal with real-time in literature. Governatori et al. have presented various partial formalisations of normative specifications with time e.g.~\cite{governatori2007characterising,governatori2011justice,governatori2005temporalised}, using \textit{defeasible logic}. Another approach is that used in C-O Diagrams~\cite{martinez:10,MCD+11tas,martinez:11,martinez:13,DCMS14svnt} giving a formal visual representation of normative systems. C-O Diagrams have been given a semantics using timed automata~\cite{Alur1994} semantics, enabling dealing with real-time directly. In contrast with the visual approach used in C-O diagrams, Themulus~\cite{themulus,cambronero2017calculus,CLP17jurix} is a timed contract-calculus including the deontic notions of permissions, prohibitions and obligations as well as reparations in case of violations, and uses a bisimulation-based approach to enable contract comparison. Kharraz et al.~\cite{KLS21tddl} developed a timed version of a dyadic deontic logic with conditional obligations, permissions, and obligations, and with a reparation operator for representing contrary-to-duties and contrary-to-prohibitions. More recently, timed contract automata~\cite{DBLP:conf/jurix/ChircopPS22} were presented as a real-time extension of contract automata~\cite{DBLP:journals/ail/AzzopardiPSS16}. These extend contract automata with clocks, inspired by \emph{timed automata}, essentially resulting in a combination of timed and contract automata handling the temporal and deontic aspects respectively. 

The concept of deontic conflicts has long been studied in the literature~\cite{FPS09acd}. For instance, an agreement which (i) obliges a person who holds a resource to release it when another party requests it, and (ii) prohibits releasing the resource halfway through a high-priority transaction, would result in a conflict when a party is halfway through a transaction and another party requests it. The notion of conflict goes beyond prohibition vs. obligation, but also covers permissions~\cite{FPS09clan}, mutually exclusive actions~\cite{FPS09clan,DBLP:journals/ail/AzzopardiPSS16} and general environmental constraints~\cite{DBLP:conf/jurix/Pace20}. However, the notion of conflict in the context of real-time deontic logics has not, to the best of our knowledge, been explored in the literature. Real-time introduces new challenges to conflict analysis. For instance, we can extend the example above to cover real time: (i) a person who holds a resource is obliged to release it within 15 minutes of another party requesting it, and (ii) it is prohibited to release a resource halfway through a high-priority transaction. Note that, adding a third clause which states that (iii) high-priority transactions may not be started while a request is pending, the conflict would be resolved.

In this paper, we explore conflict analysis in the context of real-time deontic logics. We focus on timed contract automata, where the main issue with conflict analysis lies in the presence of norms that outlive the explicit automaton state they are triggered from. We give an automata transformation that massages the automaton into a form more easily amenable to analysis. Based on this, we further describe an analysis technique to discover conflicts. This technique is sound, but not complete. We also discuss how conflict analysis can be made complete in this context, at a further complexity price. Even a sound, but not complete conflict analysis becomes non-trivial in the real-time case, and this paper is dedicated primarily to prove the correctness of our algorithm.\footnote{For the sake of conciseness, we do not cover timeouts in this paper, but we believe that the approach we have taken can be readily extended to deal with them.}

The paper is organised as follows. In Section~\ref{s:bg} we present timed contract automata and their semantics from~\cite{DBLP:conf/jurix/ChircopPS22}. We then define the notion of conflicts for these automata in Section~\ref{s:confl}, and an algorithm to transform the automata into a form more amenable for analysis in Section~\ref{s:eph}. We then present the conflict analysis algorithm and prove its correctness in Section~\ref{s:ca}. The algorithm is illustrated through a case study in Section~\ref{s:casestudy}, and we conclude in Section~\ref{s:conclusions}.

\section{Background}\label{s:bg}

Timed contract automata combine timed automata with contract automata, with states labelled by real-time norms, transitions guarded by real-time events and real-time constraints. They allow for two types of norms: (i) persistent norms, which persist even after the state they appear in is left; and (ii) ephemeral norms, which are discarded if the state is exited. In this section we briefly present the syntax and semantics of timed contract automata, but for more details, the reader is referred to~\cite{DBLP:conf/jurix/ChircopPS22}. 

In this paper we will denote the continuous time domain as $\timetype$, and the type of clocks as $\clockstype$, all of which will run at the same rate. The existence of a global clock $\gamma \in \clockstype$ is assumed. A \emph{clock valuation}, of type $\clocksvaluations \df \clockstype \rightarrow \timetype$, gives a snapshot of the values carried by the clocks. Clock predicates ranging over $\clockspredicates \df \clocksvaluations \rightarrow \booltype$ take a clock valuation and return whether the predicate is satisfied. The notation $\shift{v}{\delta}$ is used to denote the advancement of clock values in $v$ by $\delta$. A valuation $v$ is said to exceed the latest satisfaction of a clock predicate $\tau$, written $\exceeds{v}{\tau}$, if for any non-negative progress in time $\delta\in\timetype$, the predicate is not satisfied $\lnot(\tau(v+\delta))$. The expression $v \oplus v'$ denotes the overriding of a clock valuation $v$ by $v'$ i.e. the clock valuation which returns the value given by $v'$ when defined, or that given by $v$ otherwise.

Given a set of actions $\actionstype$, timed automata also identify attempted actions $\attemptedactionstype$, denoting the enriched alphabet $\actionstype\cup\{ \attempted{a} \mid a\in\actionstype\}$. A \emph{timed trace} over a set of parties $\partiestype$ and alphabet $\actionstype$ is a finite sequence of observed events: an action with associated party and timestamp (as per the global clock): $\timedtracetype \df \seq(\partiestype \times \attemptedactionstype \times \timetype)$ with increasing timestamps.

\noindent\textbf{Definition.}
A \emph{timed contract automaton} $C$, over parties $\partiestype$, actions $\actionstype$ and that uses clocks $\clockstype$, is a tuple $\langle Q,\;q_0,\;\rightarrow,\;\textit{pers},\;\textit{eph}\rangle$ where: 
(i) $Q$ is the set of states, with $q_0\in Q$ being the initial state; 
(ii) $\rightarrow \subseteq Q \times (\partiestype \times \actionstype \times \clockspredicates \times \clocksvaluations) \times Q$ is the transition relation labelling each transition with a party and action which trigger it, a clock predicate which guards it, and a (possibly partial) clock valuation to reset any number of clocks upon taking the transition; and
(iii) $\textit{pers},\;\textit{eph}\in Q \rightarrow 2^{\normstype}$ are functions, which given a state, return the sets of persistent and ephemeral norms active when in that state. 
We will write $q \transition{p:a}{\tau}{\rho} q'$ to denote $(q,\;(p,\;a,\;\tau,\;\rho),\;q') \in \rightarrow$.

A timed contract automaton is \emph{well-formed} if (i) the global clock is never reset, i.e., if $q \transition{p:a}{\tau}{\rho} q'$, then $\gamma\notin \dom(\rho)$; and (ii) the automaton is \emph{deterministic}, i.e., an observed action only allows for one transition to fire: if $q \transition{p:a}{\tau_1}{\rho_1} q_1$ and $q \transition{p:a}{\tau_2}{\rho_2} q_2$, then either $q_1=q_2$ and $\rho_1=\rho_2$, or for any clocks valuation $v$, $\lnot (\tau_1(v)\land \tau_2(v))$. In the rest of the paper we will assume that timed contract automata are well-formed. 

We will refer to syntactic reachability of a state $q$ using the predicate $\reachable(q)$ defined as there existing a path from the initial state to $q$ in the automaton.

The semantics of timed automata are given as a transition system over timed configurations as defined below:

\noindent\textbf{Definition.} A \emph{configuration} of a timed contract automaton $M=\langle Q,\;q_0,\;\rightarrow,\;\textit{pers},\;\textit{eph}\rangle$ has type: $Q\times\clocksvaluations\times\normstype\times\normstype$. We write $\configurationtype_M$ to denote the set of all configurations, leaving out $M$ when clear from the context. The initial configuration $\textit{conf}_0$ is $(q_0,\lambda c \cdot 0, \textit{pers}(q_0),\;\textit{eph}(q_0))$.

Temporal progression of configurations upon consuming an event $(p,a,t)$ is defined using the configuration relation $\textit{conf} \xRightarrow{p:a,\;t} \textit{conf}\:'$ showing how a configuration evolves, breaking it down into (i) a temporal step $\textit{conf} \xRightarrow[\textit{temp}]{p:a,\;t} \textit{conf}\:'$; and (ii) a deontic step $\textit{conf} \xRightarrow[\textit{norm}]{p:a,\;t} \textit{conf}\:'$.

\noindent\textbf{Temporal semantics:} The temporal behaviour of a timed contract automaton is characterised by two rules. The first rule (handling a change of automaton state)  triggers when enabled, and the second (handling events leaving the automaton in the same state) is triggered if the first is not. The rule handling a change of state in the configuration is the following:

{\scriptsize
\begin{mathpar}
\inferrule*[right={$\delta = t-v(\gamma),\;\tau(\shift{v}{\delta})$}]
{ q \transition{p:a}{\tau}{\rho} q' }
{ (q,\;v,\;P,\;E) 
    \xRightarrow[\textit{temp}]{p:a,\;t} 
        (q',\;(\shift{v}{\delta}) \oplus \rho,\;P \cup \textit{pers}(q'),\;\textit{eph}(q'))
}
\end{mathpar}
}\normalsize

\noindent If the rule above does not match, the configuration remains in the same state:
\scriptsize{
\begin{mathpar}
\inferrule*[right={$\delta = t-v(\gamma)$}]
{ }
{ (q,\;v,\;P,\;E) 
    \xRightarrow[\textit{temp}]{p:a,\;t}
        (q',\;\shift{v}{\delta},\;P,\;E)
}
\end{mathpar}
}\normalsize

\noindent\textbf{Deontic semantics:} The deontic semantics of timed contract automata are based on the semantics of the individual norms as defined below in terms of characterising their violations and satisfaction: 
\scriptsize{\[\begin{array}{lclclcl}
\violated{\timednorm{P}{p}{\tau}{a}}{(p:\attempted{a},\;v)} & \df & \tau(v)
&~\qquad~& \satisfied{\timednorm{P}{p}{\tau}{a}}{(p':a',\;v)} & \df & \exceeds{v}{\tau}\\
\violated{\timednorm{F}{p}{\tau}{a}}{(p:a,\;v)} & \df & \tau(v)
&& \satisfied{\timednorm{F}{p}{\tau}{a}}{(p':a',\;v)} & \df & \exceeds{v}{\tau}\\
\violated{\timednorm{O}{p}{\tau}{a}}{(p':a',\;v)} & \df & \exceeds{v}{\tau}
&& \satisfied{\timednorm{O}{p}{\tau}{a}}{(p:a,\;v)} & \df & \tau(v)
\end{array}\]
}\normalsize

Based on these definitions, the rules for deontic transitions are given below:

\scriptsize{
\begin{mathpar}
\inferrule*[right={$\delta = t-v(\gamma)$}]
{ \exists n\in P \cup E \cdot \violated{n}{(p:a, \shift{v}{\delta})} }
{ (q,\;v,\;P,\;E) 
    \xRightarrow[\textit{norm}]{p:a,\;t} 
        \bot
}\\
\inferrule*[right={$\delta = t-v(\gamma)$}]
{ \lnot\exists n\in P \cup E \cdot \violated{n}{(p:a, \shift{v}{\delta})} }
{ (q,\;v,\;P,\;E) 
    \xRightarrow[\textit{norm}]{p:a,\;t} 
        (q,\;v,\;\activenorms(P,\;(p:a,\;\shift{v}{\delta})),\;\activenorms(E,\;(p:a,\;\shift{v}{\delta})))
}
\end{mathpar}
}\normalsize

\section{Normative Conflicts}\label{s:confl}
We can now turn to the analysis of timed contract automata for conflict discovery. We will start by characterising our notion of conflict, to enable reasoning about the correctness of the algorithm discovering them. We start with the notion of a conflict at a particular point in time. In this paper, we limit ourselves to conflicts arising from a prohibition and a permission, or a prohibition and an obligation.
\begin{definition}
    A set of norms $N$ is said to conflict at a clock valuation $v$ when there is either an obligation or a permission, and a prohibition on the same action that are both active at $v$, formally:
    \begin{align*}
        \localconflicts(N, v) \df\ & \ \ \ \exists O_{\tau}(p:a), 
        F_{\tau'}(p:a) \in N \cdot \tau(v) \wedge \tau'(v)\\
                                  & \vee \exists P_{\tau}(p:a),
                                  F_{\tau'}(p:a) \in N \cdot \tau(v) \wedge \tau'(v)
    \end{align*}
\end{definition}

This notion of conflicts corresponds to the notion of conflicts defined in the discrete time setting for contract automata \cite{DBLP:journals/ail/AzzopardiPSS16}.
Using the definition above, we can talk about the notion of conflicts arising in a timed contract automaton.

\begin{definition}\label{def:conflictfree}
A timed contract automaton $M$ is said to be \emph{conflict-free}, written $\tcaconflictfree(M)$, if all timed traces will put the automaton into a configuration with no conflict:
$$\begin{array}{l}
\forall
    \textit{ts} \in \timedtracetype,\;
    (q,\;v,\;P,\;E) \in \configurationtype_M
\st \\\quad\qquad
    \textit{conf}_0 \xRightarrow{\textit{ts}} (q,\;v,\;P,\;E) \implies
    \lnot \localconflicts(P \cup E, v)
\end{array}$$
\end{definition}

It is worth noting that we identify as a conflict any point in time in which opposing norms are in force. For example, a situation where a party is obliged to perform an action between time 0 and 10, but prohibited from doing the same action between time 0 and 5 is a conflict. A more lenient definition of conflict would have allowed this, since the obligation is still satisfiable. Both are useful notions. Here we take the former view, allowing us to capture and provide to the user more useful information about the structure their behaviour. We leave the alternative interpretation as future work.

\section{From Persistent Norms to Ephemeral Norms}\label{s:eph}

The difficulty in finding conflicts primarily lies in persistent norms. If the timed contract automata only had ephemeral norms, it would be a matter of checking for conflicts locally in the states. Persistent norms, however, outlive the state they lie in, and the key to the conflict analysis analysis we present is to transform persistent norms into ephemeral ones. In the process we also complete the automaton (such that the do-nothing semantic rule is never triggered). We first define some useful functions.

Our reduction of persistent norms to ephemeral norms is exponential in the number of norms used. This is since we consider (almost) all the possible subsets of a set of norms that can be satisfied at some point in time. Given such a subset, we replicate transitions in the original automaton using the conjunction of their condition with the timing conditions during which the norms may be satisfied, and only of those norms.

\paragraph{Abstracting active norms:} Recall the \textit{active} function, which given a set of norms $N$, an action $p:a$, and a timed valuation $v$ returns $N$ without the norms satisfied by $(p:a, v)$. We abstract this with $\textit{active}_\alpha : \mathbb{N} \times \mathbb{P} \times \mathbb{A} \rightarrow 2^\mathbb{N}$, which ignores timing constraints. Instead of returning one set, it returns a set of subsets of $N$ possibly satisfied when the action $p:a$ occurs (at any point in time). We later show that for every timed valuation $v$, this abstract set contains the concrete set of active norms.

Given an action $p:a$ and a set of norms $Ns$, we define $\activekw_\alpha(N, p:a) = \{N' \subseteq N \mid \forall O_{\tau}(p':a') \in N \cdot p':a' \neq p:a \implies O_{\tau}(p':a') \in N'\}$. Note every set in $\activekw_\alpha(N, p:a)$ is a subset of $N$ and contains every obligation in $N$ not over $p:a$. Furthermore, this latter set of obligations is the only set of norms not possibly satisfied by $p:a$ occurring at any point in time. Permissions and prohibitions, even over other actions, instead can be satisfied simply by the latest satisfaction of their corresponding timed constraint.

\begin{lemma}\label{lem:correctactive}
    For all time points $v$ and non-empty actions $p:a$, $active(N, (p:a, v)) \in \activekw_\alpha(N, p:a)$.
\end{lemma}

\paragraph{Timing conditions:} To characterise the timing condition required to discharge atomic norms, we define the function $\tc \in \mathbb{N} \rightarrow \mathbb{T}$: 
\begin{align*}
    \tc(\timednorm{O}{p}{\tau}{a}) &\df \tau\\
    \tc(\timednorm{P}{p}{\tau}{a}) &\df \lambda v . v > \max(\tau)\\
    \tc(\timednorm{F}{p}{\tau}{a}) &\df \lambda v . v > \max(\tau)
\end{align*}

We overload this for sets of norms, writing $\tc(\Ns)$ for $\{\tc(N) \mid N \in \Ns\}$. Note that for permissions and prohibitions we can only discharge the norms when the timing predicate can no longer hold. For obligations, we can only discharge them when the timing predicate holds (and the action is performed). This definition mirrors the definition of satisfaction (\textit{sat}) for the respective norms.

We will use $\tc$ to identify when a set in $\activekw_\alpha$ is the only set of norms satisfied at some time point. Essentially, given a choice of $N' \in \activekw_\alpha(N, p:a)$, we want to be able to express a timing constraint $\Tau$ that is only true when there is a $v$ s.t. $N' = \activekw(N, (p:a, v))$. For ease of exposition, we define $\Tau$ for a set of norms that are satisfied ($N_{\textit{sat}}$) and not satisfied ($N_{\neg\textit{sat}}$). 

Given two norm sets $N_{\textit{sat}}$ and $N_{\neg\textit{sat}}$ we define the timing condition required for all norms in $N_{\textit{sat}}$ to be satisfied and for all norms $N_{\neg\textit{sat}}$ not to be satisfied: $\Tau(p:a, N_{\textit{sat}}, N_{\neg\textit{sat}})$, defined as 
$\tc(N_{\textit{sat}}) \wedge \neg \bigvee tc(\mathcal{P}(N_{\neg\textit{sat}}) \cup \mathcal{F}(N_{\neg\textit{sat}}) \cup \{O_\tau(p:a) \in N_{\neg\textit{sat}}\})$. 
Note that we treat obligations differently from the other types of norms, namely we ignore obligations not over $p:a$. These are immediately not satisfied by $p:a$ occurring, and not the action they predicate over. This definition does not deal with such obligations being in $N_{\textit{sat}}$, but our use of $\Tau$ will never include these in $N_{\textit{sat}}$.

We can show that for any two non-equivalent sets in $\activekw_{\alpha}(N, p:a)$, their respective constraints are mutually exclusive. We will require this later on to show determinism of out transformation.

\begin{lemma}\label{lem:mutexcl}
    Given $N_1, N_2 \in \activekw_{\alpha}(N, p:a)$, if $N_1 \neq N_2$ then $\Tau(p:a, N \setminus N_1, N_1) \wedge \Tau(p:a, N \setminus N_2, N_2)$ is false.
\end{lemma}

The proof follows by showing that if there is a permission or prohibition $n$ in $N_1$ but not in $N_2$ (or vice versa), both $tc(n)$ and its negation are in $\Tau(p:a, N \setminus N_2, N_2)$, leading to a contradiction. Obligations are handled by the definition of $\activekw_\alpha$.

Moreover, if a norm holds at a certain point in time, then the timing condition corresponding to the norm must be satisfied. 

\begin{lemma}\label{lem:tcsat}
    $sat(n, (p:a, v))$ implies $tc(n)(v)$
\end{lemma}

The proof follows easily by case analysis.

We are now ready to define our translation. Before, recall that in the timed semantics there are two different kinds of transitions: explicit transitions and implicit transitions. Our construction treats these differently, given the different semantics of each. A small note, is that we need to be careful with self-loops: there is a semantic difference between explicit self-loops and implicit self-loops. The former reactivate the ephemeral norms of a state, while the latter do not.

\paragraph{Translation:} From an automaton $M = \langle Q, q_0, \rightarrow, \pers, \eph \rangle$ our translation creates automaton $M^+ = \langle Q^+, q^+_0, \rightarrow^+, \emptyset, \eph^+ \rangle$, defined below:
\begin{enumerate}
    \item States, $Q^+ \in Q \times 2^{\mathbb{N}} \times 2^{\mathbb{N}}$, keep track of the active norms s.t. $(q, E, P) \in Q^+$ iff $q \in Q$, $E \subseteq eph(q)$, and $P \subseteq \{N \in pers(q') | q' \in Q\}$.
    \item $q^+_0 \df (q_0, eph(q_0), pers(q_0))$,
    \item $\eph^+((q, E, P)) = E \cup P$,
    \item $\rightarrow^+$ is defined as the smallest relation constructed by the following rules:
    \begin{enumerate}
        \item \textbf{Explicit transitions}: Given a state $(q, E, P)$, and a transition from $q$ in the original automaton, we consider all possible subsets of the originally persistent norms that can be satisfied by this transition, along with their conditions for satisfaction. Based on these, we create a new transition, leaving only the norms left to be satisfied, adding norms relevant to the new state, and adding the timing conditions for satisfaction of the removed norms to the guard.

        \begin{prooftree}\scriptsize
        \AxiomC{$q \xrightarrow{p:a | \tau \mapsto \rho} q'$}
        \AxiomC{$P' \in \activekw_\alpha(P, p:a)$}
        \BinaryInfC{$(q, E, P) \xrightarrow{p:a | \tau \wedge \Tau(p:a, P \setminus P', P') \mapsto \rho} (q', eph(q'), P' \cup pers(q'))$}
        \end{prooftree}

        \item \textbf{Implicit transitions}: For each state $q$ and action $p:a$, construct a condition that captures when no corresponding transition is triggered. Further consider all possible subsets of the ephemeral and originally persistent norms that can be satisfied by performing $p:a$. Based on these, we create a corresponding new self-loop transition, with a guard capturing the implicit transition triggering and the satisfaction of the guessed norms. 
    
\begin{prooftree}\scriptsize
        \AxiomC{$\tau' = \neg \bigvee \{\tau \mid
            q \xrightarrow{p:a | \tau \mapsto \rho} q'\}$}
        \AxiomC{$E' \in \activekw_\alpha(E, p:a)$}
        \AxiomC{$P' \in \activekw_\alpha(P, p:a)$}
        \TrinaryInfC{$(q, E, P) \xrightarrow{p:a | \tau' \wedge \tau_{to} \wedge \Tau(p:a, E \setminus E', E') \wedge \Tau(p:a, P \setminus P', P') \mapsto \rho_{id}} (q, E', P')$}
        \end{prooftree}

    \end{enumerate}

\end{enumerate}

Note how this reduction maintains the determinism of the original automaton, given we just refine the transitions of $M$ further and keep track of satisfaction and activation of norms. Note that any two transitions from the same state, corresponding to the same transition in $M$, are mutually exclusive given they are generated from different norm subsets.

\begin{lemma}
    $M^+$ is deterministic.
\end{lemma}

The proof follows from Lemma~\ref{lem:mutexcl}, given that we are taking the conjunction of the $\Tau$ condition with each guess of next active norms.

Finally, we show that $M^+$ has the same semantics as $M$. First we focus on the normed semantics.
Violation in $M$ implies violation in $M^+$.

\begin{lemma}\label{lem:viol}
    $(q, v, P, E) \xrightarrow[norm]{p:a, t} \perp$ implies $((q, E, P), v, \emptyset, E \cup P) \xrightarrow[norm]{p:a, t} \perp$
\end{lemma}
\begin{proof}
    The assumptions implies $\exists n \in P \cup E \cdot vio(n, (p:a, \shift{v}{\delta}))$. This is exactly the condition required to conclude the right-hand side.
\end{proof}

When there is no violation, there is a direct correspondence between the next configurations in each automaton.

\begin{lemma}\label{lem:sat}
    $(q, v, P, E) \xrightarrow[norm]{p:a, t} (q, v, P', E')$ implies $((q, E, P), v, \emptyset, E \cup P) \xrightarrow[norm]{p:a, t} ((q, E, P), v, \emptyset, E' \cup P')$
\end{lemma}

The proof follows from the definition of the norm transition semantics and distributivity of $active$.

We define $\textit{conf}^+_0 \df (q^+_0, v_0, \emptyset, eph(q^+_0)$.

\begin{theorem}\label{thm:semanticscorres}
    For all timed traces $ts$, there are sets of timed norms P and E such that $\textit{conf}_0 \xRightarrow{ts}_{M} (q, v, P, E)$ iff $\textit{conf}^+_0 \xRightarrow{ts}_{M^+} ((q, E, P), v, \emptyset, E \cup P)$.
\end{theorem}

The proof follows by induction on the length of $ts$ and the two previous lemmas.

Two immediate corollaries are that the two automata are equivalent with respect to violation and conflict-freedom. Both immediately follow from the one-to-one correspondence of active persistent and ephemeral norms between configurations of $M$ and $M^+$.

\begin{corollary}
    A timed contract automaton is violated if and only if its flattening is violated.
\end{corollary}

\begin{corollary}\label{lem:conflictmaintained}
    A timed contract automaton is conflict-free if and only if its ephemeral flattening is conflict-free: $\forall M \st 
\tcaconflictfree(M)
\iff
\tcaconflictfree(M^+)$.
\end{corollary}

\noindent\textbf{Complexity} Given the definition and use of $\activekw_\alpha$ the construction is exponential in the maximum sum, given a state in $M$, of: (1) the number of persistent obligations over the same action and with different timing predicates; and (2) the number of permissions and prohibitions.

\section{Conflict Analysis}\label{s:ca}

Recall that conflict-freedom (Defn.~\ref{def:conflictfree}) is a semantic property, requiring there is no induced configuration with conflicting norms.
The reduction to ephemeral norms maintains the same semantics as the original automaton, however it gives an automaton that corresponds more tightly to the semantics. In fact, for each runtime configuration, there is a corresponding induced state in $M^+$ that is labelled with the same norms. 

 This allows to give a simple sound algorithm to find conflicts in $M^+$, by simply analysing states for potential conflicts, and checking for satisfiability of their predicates. That is, for each state in $M^+$, consider the set of ephemeral norms $E$ and $P$ associated with it and check for satisfiability of the formula $\textit{local-conflict}(E \cup P) \df \exists v \cdot \localconflicts(E \cup P, v)$. 
 
If we find no such local conflict in an ephemeral reduction then we can conclude that $M$ is conflict-free.

\begin{theorem}
    Given a timed automata $M$, if there is no local conflict in $M^+$, then $M$ is conflict-free: If for all states $(q, E, P)$ in $Q_{M^+}$ $\neg(\textit{local-conflict}(E \cup P))$ holds, then $\textit{conflict-free(M)}$. 
\end{theorem}

This approach is sound but not complete, since some conflicting states may not be reachable. We can make the reduction finer by pruning away transitions with unsatisfiable conditions, however this is not sufficient for completeness. For example, consider a conflicting state that has an obligation that is required to be satisfied before the conflicting norms are activated, while satisfying the obligation implies leaving the state. Moreover, the state may be only possibly visited after at least one of the timing conditions of the conflicting norms no longer can hold. 

We believe a sound and complete conflict analysis may be possible, through further transforming the ephemeral flattening into a timed safety automaton with special sink states denoting conflicts. Then conflict analysis can be reduced to reachability of timed automata, which is PSPACE-complete \cite{bouyer2010model}. Given we are already paying an exponential cost to reduce persistent norms into ephemeral norms, here we prefer to rely on the sound algorithm given it can capture all the possible conflicts.

\section{Case Study}\label{s:casestudy}

To illustrate our approach we consider the case study discussed in Section~\ref{s:intro}. Consider the following two clauses: 
\begin{enumerate}[(i)]
    \item Party $A$ is obliged to release a resource within 15 minutes of party $B$ requesting it (if they hold it); 
    \item It is forbidden for party $B$ to request a resource during a high-priority transaction.  
    \item It is prohibited for party $A$ to release a resource during a high-priority transaction.
\end{enumerate}

Fig.~\ref{fig:aut} illustrates part of the corresponding automaton. Actions \textit{get}, \textit{release} and \textit{request} correspond to getting, releasing and requesting the resource in question, while actions \textit{start} and \textit{end} correspond to starting and ending a high-priority transaction. From state $q_1$ party $A$ can gain access to the resource and go to $q_2$ from which they may release it and return to $q_1$. As long as $A$ is holding the resource (in state $q_2$), $B$ may request it, upon which clock $t$ is reset, to be used by the persistent obligation in $q_3$ imposing when the resource is to be released. At this point $A$ may release the resource and return to $q_1$, or start a high-priority transaction and transition to $q_4$, wherein they are (ephemerally) prohibited from releasing the resource or for $B$ to request it. this prohibition ends when the transaction ends and the automaton transitions to $q_5$. From $q_5$, $A$ may start another transaction (and move to $q_4$), or release the resource (and move to $q_1$). 

Due to the lack of space, we do not show the ephemeral flattening. Note that the flattening, after we prune unsatisfiable transitions, will only affect state $q_4$ and $q_5$. In fact, it will add to state $q_4$ the persistent obligation in $q_3$, and it will also add it to state $q_5$. 
The only potential for conflicts is then the version of $q_4$ in the flattening, which contains both the obligation (within 15 minutes) and the prohibition from releasing the resource. Through the procedure described in Section~\ref{s:ca} we can easily detect this possible conflict (note the timing conditions of both norms are satisfiable together).

Detecting this potential conflict allows party A to \emph{a priori} adjust their behaviour accordingly, e.g., by ensuring they finish the transaction in time to still be able to satisfy the obligation, and then release the resource. If we took a weaker view of conflicts, e.g., an automaton has a conflict if there is an unavoidable conflicting configuration, we would not be able to identify such states and give these kinds of useful information to the user. Adding a prohibition on party $A$ from starting a high-priority transaction while a request is in place (i.e. in state $q_3$) would resolve this conflict.

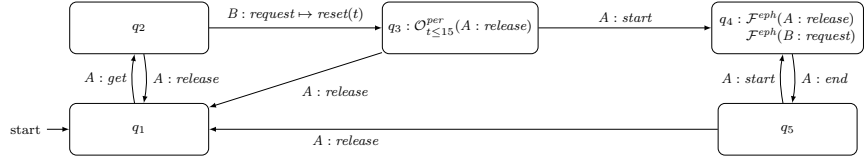
\begin{figure}
\begin{center}
\scalebox{0.65}{\makebox{
\begin{tikzpicture}[
  >=latex,
  node distance=1cm,
  auto,
   every state/.style={rectangle,rounded corners,draw,font=\footnotesize,minimum height=30pt, minimum width=80pt},
   every node/.style={font=\footnotesize},
  bend angle=90,
  ] 
\node[initial,state] (A) 
   {$q_1$}; 
\node[state] (B) [above=of A] 
   {$q_2$}; 
\node[state] (a) [right=35mm of B] 
   {$q_3: \timednormper{O}{\textit{A}}{t\leq 15}{\textit{release}}$}; 
\node[state] (b) [right=35mm of a] 
   {$\begin{array}{ll}
        q_4:& \timednormeph{F}{\textit{A}}{}{\textit{release}}\\
            & \timednormeph{F}{\textit{B}}{}{\textit{request}}
      \end{array}
   $}; 
\node[state] (bb) [below=of b] 
   {$q_5$}; 

\path[->] 
    (A) edge [bend left = 10] node {$\textit{A}: \textit{get}$} (B)
    (B) edge [bend left = 10] node {$\textit{A}: \textit{release}$} (A)
    (B) edge [] node {$\textit{B}: \textit{request}\mapsto \textit{reset}(t)$} (a)
    (a) edge [] node {$\textit{A}: \textit{release}$\qquad\ } (A)
    (a) edge [] node {$\textit{A}: \textit{start}$} (b)
    (b) edge [bend left = 10] node {$\textit{A}: \textit{end}$} (bb)
    (bb) edge [bend left = 10] node {$\textit{A}: \textit{start}$} (b)
    (bb) edge [] node {$\textit{A}: \textit{release}\qquad\qquad\qquad\qquad\qquad\qquad\qquad\qquad$} (A)
    ;

\end{tikzpicture}
}}
\end{center}
    \caption{Part of automaton corresponding to the case study constraints.}
    \label{fig:aut}
\end{figure}

\section{Conclusions}\label{s:conclusions}
In this paper, we have presented an algorithm for the discovery of conflicts in clock reset-free timed contract automata, together with a proof of its correctness. The algorithm is sound, but not complete. Real-time logics push various analyses beyond the reach of the computable and we intend to explore further timed contract automata, investigating both algorithms to perform certain analyses, and proving impossibility results. 
 
\bibliographystyle{plain}
\bibliography{refs}

\input{proofs}
\end{document}

%% file: preamble.tex
\usepackage{colortbl}
\usepackage{bussproofs}
\usepackage{tikz}
\usetikzlibrary{shapes.multipart}
\usetikzlibrary{positioning}
\usepgflibrary{shapes.gates.logic.US}
\usetikzlibrary{shapes.gates.logic.US}
\usetikzlibrary{shadows}
\usetikzlibrary{arrows}
\usetikzlibrary{trees}
\usetikzlibrary{automata}

\usepackage{xspace,mathpartir}
\usepackage[draft,inline,index,marginclue,nomargin]{fixme}
\usepackage{graphicx, xspace, amssymb, latexsym, amscd}
\usepackage{amsmath, amsbsy}
\usepackage{amsthm}
\usepackage{mathtools}
\definecolor{Blue}{rgb}{0.3,0.3,0.9}

\newcommand{\df}{\stackrel{{\tiny \textit{df}}}{=}}
\newcommand{\st}{\;\cdot\;}

\newtheorem{definition}{Definition}

\newtheorem{theorem}{Theorem}
\newtheorem{lemma}{Lemma}
\newtheorem{corollary}{Corollary}

\newcommand{\conflictsymbol}{\maltese}
\newcommand{\localconflicts}{\;\conflictsymbol\;}
\newcommand{\tcaconflictfree}{\textit{conflict-free}}

\makeatletter
\newcommand*{\textol}[1]{$\overline{\hbox{#1}}\m@th$}
\makeatother

\newcommand{\Tau}{\mathrm{T}}

\usepackage[cmtip,all]{xy}
\newcommand{\longsquiggly}{\xymatrix{{}\ar@{~>}[r]&{}}}
\newcommand{\longto}{\xymatrix{{}\ar@{->}[r]&{}}}

\newcommand{\seq}{\textit{seq}}
\newcommand{\dom}{\textrm{dom}}

\newcommand{\booltype}{\mathbb{B}}

\newcommand{\timetype}{\mathbb{T}}
\newcommand{\partiestype}{\mathbb{P}}
\newcommand{\normstype}{\mathbb{D}}
\newcommand{\actionstype}{\mathbb{A}}
\newcommand{\attemptedactionstype}{\actionstype_{\textit{attempted}}}
\newcommand{\clockstype}{\mathbb{C}}
\newcommand{\configurationtype}{\textit{Conf}}

\newcommand{\attempted}[1]{#1_{\textit{attempt}}}

\newcommand{\shift}[2]{#1 \gg #2}
\newcommand{\exceeds}[2]{#1 > \max(#2)}

\newcommand{\timednorm}[4]{{\cal #1}_{#3}(#2:#4)}
\newcommand{\timednormeph}[4]{{\cal #1}^{\textit{eph}}_{#3}(#2:#4)}
\newcommand{\timednormper}[4]{{\cal #1}^{\textit{per}}_{#3}(#2:#4)}
\newcommand{\violated}[2]{\textit{vio}(#1,\;#2)}
\newcommand{\satisfied}[2]{\textit{sat}(#1,\;#2)}
\newcommand{\activenorms}{\textit{active}}

\newcommand{\transition}[3]{\xrightarrow{#1 \;\mid\; #2 \;\mapsto\; #3}}

\newcommand{\clocksvaluations}{\textit{val}_{\clockstype}}
\newcommand{\clockspredicates}{\textit{pred}_{\clockstype}}

%% file: proofs.tex
\appendix

\section{Proofs}

\setcounter{lemma}{0}
\setcounter{theorem}{0}

\begin{lemma}
    Given $N_1, N_2 \in \activekw_{\alpha}(N, p:a)$, if $N_1 \neq N_2$ then $\Tau(p:a, N \setminus N_1, N_1) \wedge \Tau(p:a, N \setminus N_2, N_2)$ is false.
\end{lemma}

\begin{proof}
    Suppose $N_1 \neq N_2$ then, without loss of generality, assume there is a norm $n \in N_1$ but $n \not\in N_2$. 
    
    If $n$ is a permission or prohibition then there is a conjunct in $\Tau(p:a, N \setminus N_1, N_1)$ corresponding $\neg \tc(n)$, while there is a conjunct in $\Tau(p:a, N \setminus N_2, N_2)$ corresponding to $tc(n)$, ensuring the contradiction.

    Note that the definition of $\activekw_\alpha$ ensures both $N_1$ and $N_2$ contain all obligations not over $p:a$, ensuring no such obligation is in their difference. If it is over $p:a$ the same argument applies. 
\end{proof}

\begin{lemma}
    $sat(n, (p:a, v))$ implies $tc(n)(v)$
\end{lemma}

\begin{proof}
    We proceed by case analysis:
    When $n = \timednorm{O}{p}{\tau}{a}$, the LHS being true implies that $\tau(v)$, while $tc(n) = \tau$, ensuring the result.
     When $n = \timednorm{P}{p'}{\tau}{a'}$ or $n = \timednorm{F}{p'}{\tau}{a'}$, the LHS being true implies that $v > \max(\tau)$, which is equal to $\tc(n)(v)$.
    For the other cases, $\textit{sat}$ is always false by its definition. 
\end{proof}

\begin{lemma}
    $M^+$ is deterministic.
\end{lemma}

\begin{proof}
    Choose a state $(q, E, P)$. 
    
    We claim all explicit transitions from this state are mutually exclusive: note this follows from Lemma~\ref{lem:mutexcl}, given we are conjuncting the $\Tau$ condition with each guess of next active norms.

    Moreover, all transitions generated from previously implicit transitions are mutually exclusive to the explicit transitions, given $\tau'$ collects and negates the disjunction of the original timing constraints of the original explicit transitions. By Lemma~\ref{lem:mutexcl} they are also mutually exclusive with each other.

\end{proof}

\begin{lemma}
    $(q, v, P, E) \xrightarrow[norm]{p:a, t} (q, v, P', E')$ implies $((q, E, P), v, \emptyset, E \cup P) \xrightarrow[norm]{p:a, t} ((q, E, P), v, \emptyset, E' \cup P')$
\end{lemma}
\begin{proof}    
    By the definition of the norm semantics, the LHS implies $\neg\exists n \in P \cup E \cdot vio(n, (p:a, \shift{v}{\delta}))$, which allows us to conclude that $((q, E, P), v, \emptyset, E \cup P) \xrightarrow[norm]{p:a, t} ((q, E, P), v, \emptyset, active(E \cup P, (p:a, v))$.

    Note $active(E \cup P, (p:a, v))$ is equal to $active(E, (p:a, v)) \cup active(P, (p:a, v))$, given the distributivity of $active$. This is exactly $E' \cup P'$, given by definition of the norm semantics (for the LHS) $P' = active(P, (p:a, v))$ and $E' = active(E, (p:a, v))$.
\end{proof}

\begin{theorem}
    For all timed traces $ts$, there are sets of timed norms P and E such that $\textit{conf}_0 \xRightarrow{ts}_{M} (q, v, P, E)$ iff $\textit{conf}^+_0 \xRightarrow{ts}_{M^+} ((q, E, P), v, \emptyset, E \cup P)$.
\end{theorem}
\begin{proof}    
    We show this by induction on the length of $ts$. For the base case, consider that the initial configuration of $M$ is $(q_0, v_0, \pers(q_0), \eph(q_0))$. By definition of $M^+$ its initial configuration is $((q_0, \eph(q_0), \pers(q_0)), v_0, \emptyset, \eph(q_0) \cup \pers(q_0))$, thus the result follows. The inductive step remains.
    
    Consider a configuration $(q, v, P, E)$ of $M$ and the corresponding configuration (given by the inductive hypothesis) in $M^+$, $((q, E, P), v, \emptyset, E \cup P)$. Consider also a timed action $(p : a, t)$. It remains to show that after this step the theorem is maintained. There are three cases, either an explicit, or implicit transition activates, or there is a violation.

    Case 1: The case of a violation is Lemma~\ref{lem:viol}.

    Case 2: Assume the timed action $(p:a, v)$ occurs, and there is a corresponding transition triggered in $M$, $q \xrightarrow{p : a \mid \tau \mapsto \rho} q'$.

    The induced deontic semantics transition in $M$ is then $(q, v, E, P) \xrightarrow[norm]{p:a, t} (q, v, active(P, (p:a, \shift{v}{\delta})), active(E, (p:a, \shift{v}{\delta})))$. 

    The induced timed configuration in $M$ is then $(q, v, active(P, (p:a, \shift{v}{\delta})), active(E, (p:a, \shift{v}{\delta}))) \xrightarrow[temp]{p:a, t} (q', (\shift{v}{\delta}) \oplus \rho, active(P, (p:a,\shift{v}{\delta})) \cup \pers(q'), \eph(q'))$.

    The induced deontic semantics transition in $M^+$ is then $((q, E, P), v, \emptyset, E \cup P) \xrightarrow[norm]{p:a, t} (q, v, \emptyset, active(E \cup P, (p:a, \shift{v}{\delta})))$, by Lemma~\ref{lem:sat}.

    By definition of $M^+$, for each $P' \in \activekw_\alpha(P, p:a)$, it has a transition: $(q, E, P) \xrightarrow{p:a | \tau \wedge \Tau(p:a, P \setminus P', P') \mapsto \rho} (q', eph(q'), P' \cup pers(q'))$. 

    By Lemma~\ref{lem:correctactive}, there is such a transition for $active(P, (p:a, \shift{v}{\delta}))$ -- onwards we set $P' = active(P, (p:a, \shift{v}{\delta}))$. What remains is to show that the timing condition $\Tau(p:a, P \setminus P', P')$ holds true at the time point $t$ that $p:a$ occurs at.

    $\Tau(p:a, P \setminus P', P')$ is made of two conjuncts: $\tc(P \setminus P')$ and $\neg \bigvee \tc(\mathcal{P}(P \setminus P') \cup \mathcal{F}(P') \cup \{O_\tau(p:a) \in N_{P'}\})$.

    Recall that by definition of $active(P, (p:a, \shift{v}{\delta}))$, $sat$ holds true for each norm in $P \setminus active(P, (p:a, \shift{v}{\delta}))$ at $v$. Since $\tc$ mirrors the definition of $sat$, then $\tc(P \setminus active(P, (p:a, \shift{v}{\delta})))(\shift{v}{\delta})$ holds true.
    
    Dually, $sat$ is not true on any norm in $active(P, (p:a, \shift{v}{\delta}))$ at $\shift{v}{\delta}$, ensuring that the second conjunct also holds true at $\shift{v}{\delta}$ (by Lemma~\ref{lem:tcsat}).
    
    Case 3: Note, that $\tau'$ in the implicit transitions rule identifies the case that there is no explicit transition for $p:a$. The rest of the proof follows similarly to the previous case. 
\end{proof}